\documentclass[conference]{IEEEtran}

\usepackage{cite}

{

\ifCLASSINFOpdf
  \usepackage[pdftex]{graphicx}
\else
  \usepackage[dvips]{graphicx}

\fi
\usepackage{epstopdf}
\usepackage{amsmath}
\usepackage{amsthm}
\usepackage{stmaryrd}
\usepackage{amsfonts}
\usepackage{algorithmic}
\usepackage{algorithm}
\usepackage{mathrsfs}
\usepackage{array}
\usepackage{color}

\hyphenation{op-tical fo-llowing net-works semi-conduc-tor}

\newtheorem{theorem}{Theorem}
\newtheorem{lemma}{Lemma}

\newcolumntype{M}[1]{>{\centering\arraybackslash}m{#1}}

\title{A New Lower Bound on the Ergodic Capacity of Optical MIMO Channels}

\author{\IEEEauthorblockN{R\'emi Bonnefoi, Amor Nafkha}
\IEEEauthorblockA{CentraleSup\'elec/IETR, CentraleSup\'elec Campus de Rennes, 35510 Cesson-S\'evign\'e, France\\
Email:\{remi.bonnefoi, amor.nafkha\}@centralesupelec.fr}}

\begin{document}
\maketitle

\begin{abstract}
In this paper, we present an analytical lower bound on the ergodic capacity of optical multiple-input multiple-output (MIMO) channels.  It turns out that the optical MIMO channel matrix which couples the $m_t$ inputs (modes/cores) into $m_r$ outputs (modes/cores) can be modeled as a sub-matrix of a $m \times m$ Haar-distributed unitary matrix where $m > m_t, m_r$. Using the fact that the probability density of the eigenvalues of a random matrix from unitary ensemble can be expressed in terms of the Christoffel-Darboux kernel. We provide a new analytical expression of the ergodic capacity as function of signal-to-noise ratio (SNR). Moreover, we derive a closed-form lower-bound expression to the ergodic capacity. In addition, we also derive an approximation to the ergodic capacity in low-SNR regimes. Finally, we present numerical results supporting the expressions derived.
\end{abstract}

\section{Introduction}

With the advent of massive multiple-input multiple-output (MIMO) technologies and the development of Internet of Things (IoT), the fifth generation cellular networks (5G) should be supported by a high quality back-haul. This back-haul can be realized through both wired and wireless technologies. Among all the possible solutions, the deployment of optical fibers is the one that ensures the greatest throughput while maintaining a high level of reliability. Moreover, space-division multiplexing (SDM) based on multi-core/multi-mode optical fiber can significantly increase the capacity limit of optical fibers \cite{MMF,SDM} and it overcomes the capacity crunch. To take full advantage of the potential throughput increase promised by SDM, the in-band crosstalk must be properly managed. This can be done using MIMO signal processing techniques. Moreover, assuming negligible back scattering and near loss-less propagation,  the propagation channel in an optical fiber can be modeled by a complex random unitary matrix.

For wireless communications, the Rayleigh fading model \cite{Telatar} is used to model the MIMO channel, in this case, the entries of the channel matrix can be modeled by an independent and identical distributed zero mean Gaussian complex numbers. Moreover, the matrices $\textbf{H}^{\dag}\textbf{H}$ are referred to as uncorrelated Wishart matrices, where $(.)^{\dag}$ is the complex trans-conjugate. In optical fiber, the channel matrix can be modeled by a Haar distributed matrix \cite{Winzer}. Therefore, the matrices follow the Jacobi unitary ensemble. In \cite{Dar} the ergodic capacity of the Jacobi MIMO channel was expressed as an integral and sum of Jacobi polynomials. 

In this article, the  ergodic capacity formula is first rewritten to show that the polynomial part of the integrand consists of the Christoffel-Darboux kernel \cite{Andrews}. Then, the Christoffel-Darboux formula is used to reword the expression of the ergodic capacity of loss-less SDM optical fiber channel. Finally, the derived new expression is used to propose a lower-bound of the ergodic capacity. We finally derive a low SNR approximation of the capacity. According to \cite{Winzer} and \cite{Dar} as well as authors knowledge, no existing work addresses the lower-bound and low-SNR approximation of the ergodic capacity of the Jacobi MIMO channel.

The rest of this paper is organized as follows. The system model is introduced in Section II. In Section III, the Christoffel-Darboux formula is used to derive a new expressions for the ergodic capacity and its lower bound. Some numerical results are provided to validate the accuracy of the derived expressions in Section IV. Finally, Section V provides a conclusion.

\section{System Model}

In an optical MIMO-SDM system with $m_t$ transmit cores/modes and $m_r$ receiving cores/modes, the expression of the received signal can be expressed as \cite{Winzer}:

\begin{equation}
\textbf{y} =\textbf{H}\textbf{x}+\textbf{n}
\end{equation}

Where $\textbf{y}\in \mathbb{C}^{m_r \times 1}$ is the vector of received symbols, $\textbf{x}\in \mathbb{C}^{m_t \times 1}$ vector of transmitted symbols, $\textbf{n}\in \mathbb{C}^{m_r \times 1}$ is the Gaussian with zero mean and unitary variance noise vector and $\textbf{H}\in \mathbb{C}^{m_r \times m_t}$ is the channel matrix.

A multi-modes/multi-cores optical fiber composed by $m$ modes/cores can be modeled by a $m \times m$ matrix denoted $\textbf{G}$. In the case of $m_t,\, m_r \leq m$, if $m_t$ ($m_r$) modes/cores are excited for transmission (reception), the channel matrix $\textbf{H}$ is a $m_r \times m_t$ block of $\textbf{G}$. Without loss of generality, we can take the upper-left corner of $\textbf{G}$ \cite{Collins}. In optical fiber, $\textbf{H}$ can be modeled as a random Haar distributed matrix \cite{Dar}. This matrix can be achieved by orthonormalizing an independent identically distribution random Gaussian matrix. Moreover, if $\textbf{H}$ is a Haar distributed matrix, $\textbf{H}^{\dag}\textbf{H}$ follows the Jacobi Unitary Ensemble (JUE).

Let us recall the definition of  the Jacobi polynomials as follows. The Jacobi polynomial of degree $n$ is denoted $P_n^{a,b}(x)$. Those polynomials are orthogonal with respect to the inner product:

\begin{equation}
\langle P(x)|Q(x)\rangle =\int_{-1}^1 (1-x)^{a} (1+x)^{b} P(x)Q(x) \; dx
\end{equation}

More generally, with this inner product a Jacobi polynomial of degree $n$ is orthogonal with all polynomials of degree $n-1$ or lower.

Furthermore, we suppose that the receiver has complete channel state information (CSI) and that the transmitter has no CSI. In this case, the definition of the ergodic capacity of the channel is given by \cite{Telatar, Winzer, Dar}:

\begin{equation}
\label{eq:capdef}
C_{m_t,m_r}^{m,\rho} = \mathbb{E}_\textbf{H} \left\{ \log_2\left(\det(\textbf{I}_{m_t}+\rho \textbf{H}^{\dag}\textbf{H}) \right) \right\}
\end{equation}

Where $\mathbb{E}_\textbf{H}$ is the expectation with respect to the density of the eigenvalues of $\textbf{H}^{\dag}\textbf{H}$  and $\rho$ the signal-to-noise-ratio (SNR) which can be defined by the ratio between the variance of the received power and the variance of the noise.

When the total number of antennas exceeds the number of modes/cores in the fiber, $m_r+m_t>m$, the ergodic capacity can be expressed as \cite{Dar}:

\begin{equation}
\label{eq:capmsup}
C_{m_t,m_r}^{m,\rho} = (m_t+m_r-m)\log_2(1+\rho)+C_{m-m_r,m-m_t}^{m,\rho}
\end{equation}

Thus, without loss of generality, in this paper, our study will be limited to the case $m_r+m_t\leq m$, and the expression of the ergodic capacity is given by \cite{Dar}

\begin{multline}
\label{eq:cap1}
	C_{m_t,m_r}^{m,\rho} = \int_0^1 \lambda^{a} (1- \lambda)^{b} \log_{2}(1 + \lambda \rho)\\
\times \sum_{k=0}^{r-1} B_{k,a,b}^{-1} \left[P_{k}^{a,b}(1-2 \lambda) \right]^{2}d\lambda
\end{multline}

	Where $P^{a,b}_{k}(x)$ is a Jacobi polynomial of degree $k$ and:

\begin{align}	
	B_{k,a,b} & = \frac{\vert \vert P_{k}^{a,b}(x)\vert\vert^2}{2^{a+b+1}} \nonumber\\
	& = \frac{1}{2k+a+b+1} {2k+a+b\choose k}{2k+a+b\choose k+a}^{-1} 
\end{align}

Where $ {n \choose k}=\frac{n!}{k!(n-k)!}$ denotes the binomial coefficient, and $a = |m_{r}-m_{t}|$, $b = m-m_{r}-m_{t}$, $r = min\{m_{r},m_{t}\}$

\section{New Expression of the Ergodic Capacity}

Moreover, the expression of the diagonal Christoffel-Darboux kernel is given by \cite{Andrews}:

\begin{equation}
		K_{n}(x,x) = \sum_{k=0}^{n}\frac{1}{\|p_{k}\|^{2}}p_{k}(x)^2
\end{equation}

And the Christoffel-Darboux formula:

\begin{equation}
\label{eq:CRformula}
	K_n(x,x)=\frac {k_n}{k_{n+1}\|p_{n}\|^{2}}\ (p'_{n+1}(x)p_n(x)-p'_n(x)p_{n+1}(x))
\end{equation}

Where $k_n$ is the leading coefficient of the orthogonal polynomial $p_n$. 

With the variable change $x = 1-2\lambda$, the expression of the ergodic capacity \eqref{eq:cap1} can be rewritten:

\begin{multline}
\label{eq:cap2}
	C_{m_t,m_r}^{m,\rho} = \int_{-1}^{1} (1-x)^{a} (1+x)^{b}\log_{2}\left(1 +\frac{\rho(1-x)}{2}\right) \\
	\times\sum_{k=0}^{r-1}  
\frac{\left[P_{k}^{a,b}(x) \right]^{2}}{\|P_{k}^{a,b}\|^{2}}dx
\end{multline}

The first theorem of this paper is obtained by applying the Christoffel Darboux formula for Jacobi polynomials on \eqref{eq:cap2} and then by computing the derivatives of the polynomials.

\begin{theorem}
\label{th:1}
For $m_r+m_t\leq m$ with perfect CSI at the receiver and no CSI at the transmitter, the expression of the ergodic capacity of a Jacobi MIMO channel is:

\begin{multline}
\label{eq:cap3}
	C_{m_t,m_r}^{m,\rho} = M_{r}^{a,b} \int_{-1}^{1} (1-x)^{a} (1+x)^{b}\log_{2}\left(1 +\frac{\rho(1-x)}{2}\right)\\
\times \left[ P_{r-1}^{a,b}(x)P_{r-1}^{a+1,b+1}(x)-N_{r}^{a,b}P_{r}^{a,b}(x)P_{r-2}^{a+1,b+1}(x) \right]dx
\end{multline}

Where, $M_{r}^{a,b} = \frac{(r+a+b+1)!r!}{2^{a+b+1}(r+a-1)!(r+b-1)!(2r+a+b)}$, and $N_{r}^{a,b}=\frac{({r}+a+b)}{({r}+a+b+1)}$.

\end{theorem}

In case where $r$ is large ($r\geq 3$), this new expression is more convenient than \eqref{eq:cap1} for numerical evaluations. Indeed, the computation of the capacity with \eqref{eq:cap1} requires the computation of the sum of $r-1$ products of Jacobi polynomials, whereas, in \eqref{eq:cap3} this sum is replaced by the difference between two products.

The new expression expressed in theorem \ref{th:1} can be used to derive a lower bound for the capacity:

\begin{lemma}
\label{lem:1}
If $m_r+m_t \leq m$, the ergodic capacity of the Jacobi MIMO channel is lower bounded by:

\begin{multline}
\label{eq:lb}
	Lb_{m_t,m_r}^{m,\rho} = M_{r}^{a,b} \int_{-1}^{1} (1-x)^{a} (1+x)^{b}\log_{2}\left(1 +\frac{\rho(1-x)}{2}\right)\\
	 \times P_{r-1}^{a,b}(x)P_{r-1}^{a+1,b+1}(x)dx
\end{multline}
\end{lemma}

\begin{proof}
To prove that \eqref{eq:lb} is a lower-bound of the capacity, it is sufficient to prove that the quantity:

\begin{multline}
Q = \int_{-1}^{1} (1-x)^{a} (1+x)^{b}\log_{2}\left(1 +\frac{\rho(1-x)}{2}\right)\\
	\left[ P_{r}^{a,b}(x)P_{r-2}^{a+1,b+1}(x) \right]dx
\end{multline}

is negative.

First, the use of the inequality of convexity of the logarithm $\log(1+x)\leq x$ leads to:

\begin{equation}
Q \leq \frac{\rho}{2 \ln(2)} \int_{-1}^{1} (1-x)^{a+1} (1+x)^{b}\left[ P_{r}^{a,b}(x)P_{r-2}^{a+1,b+1}(x) \right]dx
\end{equation}

Moreover, the link between the degree and the coefficients of Jacobi polynomials can be made with \cite{NIST:DLMF}:

\begin{equation}
\label{eq:form1}
P_{n}^{a-1,b}(x) = \frac{(n+a+b)}{(2n+a+b)}P_{n}^{a,b}(x)-\frac{(n+b)}{(2n+a+b)}P_{n-1}^{a,b}(x)
\end{equation}

\eqref{eq:form1} and the orthogonality of Jacobi polynomials allow us to conclude that:

\begin{equation}
\int_{-1}^{1} (1-x)^{a+1} (1+x)^{b}\left[ P_{r}^{a,b}(x)P_{r-2}^{a+1,b+1}(x) \right]dx = 0
\end{equation}

This finally proves that $Q \leq 0$. This last result proves lemma \ref{lem:1}.

Moreover, at low SNR, $\log(1+x)\approx x$ and thus $Q \approx 0$, thus, at low SNR,

\begin{equation}
  Lb(m_{t},m_{r},m,\rho)\approx C(m_{t},m_{r},m,\rho)
\end{equation}
\end{proof}

We can note that, at low SNR, $\log(1+x)\approx x$ and thus $Q \approx 0$ and the proposed lower bound is a good approximation of the capacity.

The novel expression of the ergodic capacity derived in theorem \ref{th:1} can also be used to propose a low SNR approximation of the ergodic capacity. Indeed, at low SNR, the ergodic capacity can be approximated by a very simple linear expression.

\begin{lemma}
\label{lem:2}

At low SNR, the expression of the ergodic capacity of the Jacobi MIMO channel is:

\begin{equation}
C_{m_t,m_r}^{m,\rho} \approx  \frac{\rho m_rm_t}{m \ln(2)} \qquad \rho \ll 1
\end{equation}
\end{lemma}

Thus, to maximize the capacity at low SNR, the number of transmit and receiving cores/modes must be maximized and the number of unused cores/modes must be minimized.

\begin{proof}
For $m_r+m_t \leq m$,

Firstly, with the first order Taylor expansion of the logarithm, $ \ln(1+x) \underset{0}{\sim}x$, \eqref{eq:cap3} becomes:

\begin{multline}
\label{eq:capapp1}
	C(m_{t},m_{r},m,\rho) = \frac{\rho M_{r}^{a,b}}{2\ln(2)} \int_{-1}^{1} (1-x)^{a+1} (1+x)^{b}\\
 \left[ P_{r-1}^{a,b}(x)P_{r-1}^{a+1,b+1}(x)-N_{r}^{a,b}P_{r}^{a,b}(x)P_{r-2}^{a+1,b+1}(x) \right]dx
\end{multline}

It has already been proved that:

\begin{equation}
 \int_{-1}^{1} (1-x)^{a+1} (1+x)^{b}P_{r}^{a,b}(x) P_{r-2}^{a+1,b+1}(x)dx =0
\end{equation}

Then, with \eqref{eq:form1} and by performing the variable change $x = -x$, the expression of the capacity at low SNR becomes:

\begin{multline}
C(m_{t},m_{r},m,\rho) = \frac{\rho M_{r}^{a,b}}{2\ln(2)}\int_{-1}^{1}  (1-x)^{b}(1+x)^{a+1}\\
 \times\left[ A P_{r-1}^{b,a+1}(x)P_{r-1}^{b+1,a+1}(x) +B P_{r-2}^{b,a+1}(x)P_{r-1}^{b+1,a+1}(x) \right]dx
\end{multline}

Where 

\begin{equation}
A = \frac{(r+a+b)}{(2r+a+b-1)} \qquad B = \frac{(r+b-1)}{(2r+a+b-1)}
\end{equation}

Finally, we can compute this integral with \cite{integrals}:

\begin{multline}
\label{eq:form2}
\int_{-1}^1 (1-x)^{c} (1+x)^{b} P_n^{a,b}(x)P_m^{c,b}(x) \; dx = \\
\frac{2^{b+c+1}(a+b+m+n)!(b+n)!(c+m)!(a-c+n-m-1)!}{m!(n-m)!(a+b+n)!(b+c+m+n+1)!(a-c-1)!}
\end{multline}

Finally, \ref{eq:form2} proves that at low SNR, the expression of the ergodic capacity is:

\begin{equation}
C(m_{t},m_{r},m,\rho) =\frac{\rho}{\ln(2)}\frac{(a+r)r}{(a+b+2r)}
\end{equation}

Which proves lemma \ref{lem:2} in the case where $m_r+m_t\leq m$. This result can be extended to the case $m_r+m_t >m$ with the first order Taylor expansion of \eqref{eq:capmsup}.
\end{proof}

\section{Numerical results}

In this section, we provide a set of Matlab simulations that illustrate the theoretical results presented in the previous sections. We first suppose that the number of supported modes in the fiber is equal to $32$ and that the SNR varies between 0 and 30 dB. In a first step, the ergodic capacity is computed using the expression proposed in theorem \ref{th:1}, then the computation is made using \eqref{eq:cap1} proved in \cite{Dar}. The obtained results with the two formulas, for different values of $m_r$ and $m_t$, are drawn in figure \ref{fig:comp}. 

 \begin{figure}[!ht]

\centering
\includegraphics[scale=0.5]{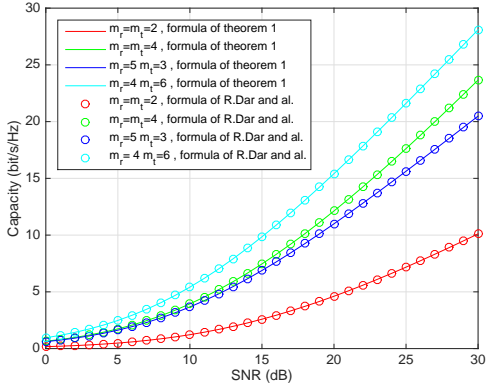}

\caption{Numerical evaluation of the ergodic capacity when $m=32$}
\label{fig:comp}
\end{figure}

Figure \ref{fig:comp} shows that the two expressions of the ergodic capacity ( Eq.\eqref{eq:cap1} and Eq.\eqref{eq:cap3}) produce the same simulation results. Furthermore, the proposed expression of the ergodic capacity reduces the evaluation time without impacting the results. Note that $m_r$ and $m_t$ have a symmetrical role in the expression of the ergodic capacity. Also, the evaluation time is drastically reduced when the parameter $r = min\{m_{r},m_{t}\}$ is large.

 \begin{figure}[!ht]

\centering
\includegraphics[scale=0.5]{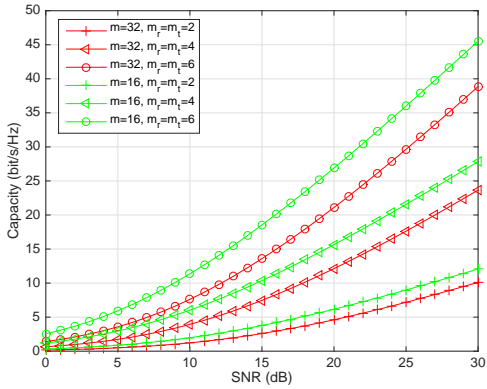}

\caption{Impact of unused cores/modes on the ergodic capacity}
\label{fig:unused}
\end{figure}

In figure \ref{fig:unused}, we analyze the effect of unused modes on the channel ergodic capacity. We compare the evolution of the capacity where $m=16$ and $m=32$ for different values of $m_r$ and $m_t$. For a fixed number of transmit and receiving cores/modes, the channel ergodic capacity decreases as $m$ increases. In other words, the capacity decreases as the number of unused cores/modes increases. This observation allows to generalize the observation done at low SNR.

 \begin{figure}[!ht]

\centering
\includegraphics[scale=0.5]{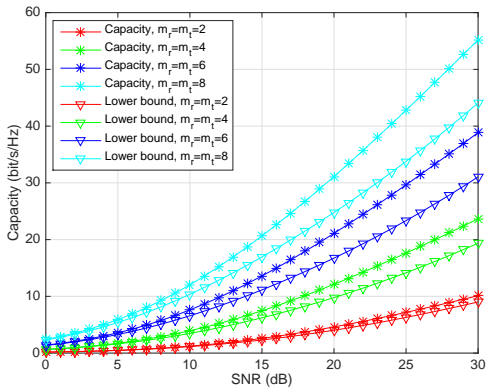}

\caption{Simulated ergodic capacity and analytical lower bound against the SNR when $m=32$ and $m_t=m_r$}

\label{fig:lb}
\end{figure}

Using the same simulation parameters ($m$, $m_r$, $m_t$ and the SNR variation range) as in figure \ref{fig:comp}, the channel ergodic capacity and its lower-bound are depicted against the SNR in figure \ref{fig:lb}. As anticipated, at low SNRs, the lower-bound of the ergodic capacity is very close to the ergodic capacity. However, them increases as  the capacity gap between them increases when SNR is large.

\begin{figure}[!ht]
\centering
\includegraphics[scale=0.5]{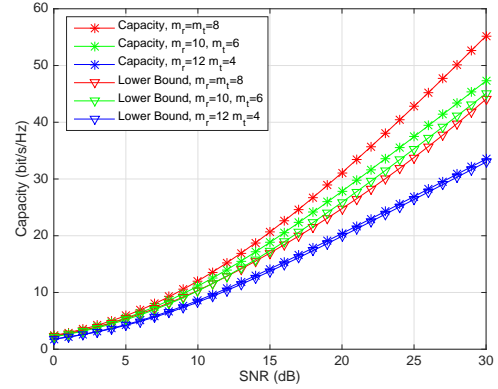}
\caption{Simulated ergodic capacity and analytical lower bound against the SNR when $m=32$ and $m_t \neq m_r$ }
\label{fig:diff}
\end{figure}

In figure \ref{fig:diff}, we analyze the evolution of the lower-bound of the ergodic capacity for different values of $m_r$ and $m_t$ when $m_r+m_t=16$ and $m=32$. The gap between the ergodic capacity and the proposed lower bound expression decreases as $a=|m_r-m_t|$ is large. Moreover, this capacity gap reaches a maximum when $m_r=m_t$.

\section{Conclusion}

In this paper, we used the Jacobi MIMO channel to analyze the ergodic capacity of the propagation channel in a multi-mode/multi-core optical fiber. We first used the Christoffel-Darboux formula to reformulate the expression of the ergodic capacity. The proposed new expression reduces the computation time to evaluate the ergodic capacity. Moreover, this expression has been used to propose a lower-bound of the capacity. We finally derived a very simple low SNR approximation of the ergodic capacity.

\bibliographystyle{ieeetr}
\bibliography{biblio_letter}

\end{document}